\documentclass[11pt]{article}

\usepackage{fullpage}
\usepackage{amsmath,amsfonts,amsthm,xspace,hyperref,graphicx,relsize,bm}
\usepackage{mathpazo}
\usepackage[margin=1in]{geometry}
\usepackage{float}
\usepackage{endnotes}
\usepackage{enumitem}
\usepackage{color}
\usepackage{amssymb,latexsym}
\usepackage{multicol}
\usepackage{authblk}
\allowdisplaybreaks

\newtheorem{theorem}{Theorem}
\newtheorem{proposition}[theorem]{Proposition}
\newtheorem{lemma}[theorem]{Lemma}

\newtheorem{fact}[theorem]{Fact}

\theoremstyle{definition}

\newcommand{\ket}[1]{|#1\rangle}
\newcommand{\bra}[1]{\langle#1|}

\newcommand{\ketbra}[2]{|#1\rangle\! \langle #2|}
\newcommand{\Tr}{\mbox{\rm Tr}}
\newcommand{\Id}{\mathbb{I}}

\DeclareMathOperator*{\Ex}{\mathbb{E}}

\newcommand{\C}{\ensuremath{\mathbb{C}}}

\newcommand{\E}{\mathcal{E}}

\renewcommand{\P}{\mathsf{P}}
\newcommand{\Q}{\mathsf{Q}}

\newcommand{\eps}{\varepsilon}

\newcommand{\X}{\mathcal{X}}
\newcommand{\Y}{\mathcal{Y}}
\newcommand{\A}{\mathcal{A}}
\newcommand{\B}{\mathcal{B}}

\newcommand{\mi}{{-i}}

\newcommand{\wt}[1]{\widetilde{#1}}
\newcommand{\what}[1]{\widehat{#1}}
\newcommand{\eval}{\mathrm{val}^*}
\newcommand{\val}{\mathrm{val}}
\newcommand{\ac}{\mathrm{Z}}

\newcommand{\xvec}{{x_{[n]}}}
\newcommand{\yvec}{{y_{[n]}}}
\newcommand{\avec}{{a_{[n]}}}
\newcommand{\bvec}{{b_{[n]}}}

\newcommand{\Xvec}{X_{[n]}}
\newcommand{\Yvec}{Y_{[n]}}

\begin{document}

\title{A parallel repetition theorem for all entangled games}
\author{Henry Yuen\thanks{hyuen@mit.edu} \\ MIT }

\date{}
\maketitle

\begin{abstract}
	The behavior of games repeated in parallel, when played with quantumly entangled players, has received much attention in recent years. Quantum analogues of Raz's classical parallel repetition theorem have been proved for many special classes of games. However, for general entangled games no parallel repetition theorem was known.
	
	We prove that the entangled value of a two-player game $G$ repeated $n$ times in parallel is at most $c_G n^{-1/4} \log n$ for a constant $c_G$ depending on $G$, provided that the entangled value of $G$ is less than $1$. In particular, this gives the first proof that the entangled value of a parallel repeated game must converge to $0$ for \emph{all} games whose entangled value is less than $1$. Central to our proof is a combination of both classical and quantum correlated sampling.
 \end{abstract}

\section{Introduction}
%
%

A two-player one-round game $G$ is played between a referee and two isolated players (who we will call Alice and Bob), who communicate only with the referee and not between themselves. The referee first samples a question pair $(x,y)$ from some distribution $\mu$ and sends $x$ to Alice and $y$ to Bob. Alice and Bob respond with answers $a$ and $b$ respectively, and they win if $V(x,y,a,b) = 1$ for some predicate $V$. 

The maximum winning probability of Alice and Bob in a game $G$ is a quantity that depends on what resources they are allowed to use. If their answers are a deterministic function of their received question (and perhaps some public random string), then we call their maximum winning probability the \emph{classical value} of $G$, denoted by $\val(G)$. However quantum mechanics allows Alice and Bob to share a resource called \emph{entanglement}, which gives rise to correlations that cannot be reproduced with public randomness only. When Alice and Bob make use of entanglement to play a game $G$, we call their maximum winning probability the \emph{entangled value} of $G$, denoted by $\eval(G)$. For all games, the classical value is at most the entangled value. Cast in the language of games, the famous Bell's Theorem states that there exist games $G$ where those values are different: $\eval(G) > \val(G)$~\cite{bell1964}.

The Parallel Repetition Question is the following natural and basic question: given a game $G$ with value less than $1$, what is the value of the game $G^n$, wherein Alice and Bob play $n$ independent instances of $G$ played \emph{in parallel}? More formally, in the game $G^n$, the referee samples $n$ independent question pairs $(x_1,y_1),\ldots,(x_n,y_n)$ from $\mu$, and sends $(x_1,\ldots,x_n)$ to Alice, and sends $(y_1,\ldots,y_n)$ to Bob. Alice responds with answer tuple $(a_1,\ldots,a_n)$, Bob responds with $(b_1,\ldots,b_n)$, and the players win if for all coordinates $i \in [n]$, $V(x_i,y_i,a_i,b_i) = 1$.

The difficulty in relating $\val(G^n)$ with $\val(G)$ and $n$ is that even though each of the $n$ instances of $G$ in $G^n$ are independent, Alice and Bob need not play each instance independently. For example, since Alice receives $(x_1,\ldots,x_n)$ all at once, she can use some question $x_j$ to answer the $i$'th game, and Bob can do something similar. Because of such strategies, for every $k$ there are games $G$ such that $\val(G^k) = \val(G) < 1$. This shows that the naive expectation that $\val(G^n) = \val(G)^n$ is false. 

The naive expectation is not too far from the truth, however: Raz's Parallel Repetition Theorem~\cite{raz1998parallel} states that
$$
	\val(G^n) \leq (1 - (1 - \val(G))^3)^{c_G n},
$$ 
where $c_G$ is a constant depending on $G$. In particular, as $n$ goes to infinity, the classical success probability goes to $0$ exponentially fast in $n$ (provided that $\val(G) < 1$). The proof is highly nontrivial, although it has been simplified and improved upon in recent years~\cite{Hol09,Braverman2015}. Raz's Parallel Repetition Theorem has heavily influenced complexity theory, most notably in the areas of hardness of approximation~\cite{hastad2001} and communication complexity~\cite{jain2011,braverman2013direct}.

One open question, which we call the Quantum Parallel Repetition Conjecture, asks whether an analogue of Raz's Parallel Repetition Theorem holds in the setting of entangled players. The Quantum Parallel Repetition Conjecture has been resolved for many special cases of games, including free games~\cite{chailloux2014parallel,JainPY14,chung2015parallel}, projection games~\cite{DinurSV14}, XOR games~\cite{cleve08}, unique games~\cite{kempe2008}, anchored games~\cite{BVY15anchoring}, and fortified games~\cite{BVY15fort}. However, the general case has remained elusive. Not only do we not know of a quantum analogue of Raz's Parallel Repetition Theorem, it hasn't even been shown that if $\eval(G) < 1$, then $\eval(G^n)$ goes to $0$ as $n$ goes to infinity! Could quantum entanglement allow players to counteract the value-decreasing effect of parallel repetition?

In this paper we prove that for all nontrivial entangled games $G$ (i.e. $\eval(G) < 1$), the entangled value of $G^n$ must converge to $0$. This resolves a weaker version of the Quantum Parallel Repetition Conjecture for general games. Quantitatively, our result is the following:

\begin{theorem}[Main Theorem]
\label{thm:main}
	Let $G$ be a game involving two entangled players with $\eval(G) = 1 - \eps$. Then for all integer $n > 0$,
	$$
		\eval(G^n) \leq c \cdot \frac{s_G \log n}{\eps^{17} n^{1/4}}
	$$
	where $c$ is a universal constant and $s_G$ is the bit-length of the players' answers in $G$.
\end{theorem}
This shows that the entangled value of $G^n$ must decay at a polynomial rate with $n$. The full Quantum Parallel Repetition Conjecture states that the rate of decay is in fact exponential, and this remains an important open problem.

\subsection{Previous work} 

There has been extensive work on the parallel repetition of entangled games. As stated earlier, past results have applied to various special classes of games, but there was no result that covered \emph{all} games. 

The results coming closest to the Quantum Parallel Repetition Conjecture are the work of Kempe and Vidick~\cite{kempe2011parallel} and Bavarian, Vidick, and Yuen~\cite{BVY15anchoring,BVY15fort}. Rather than proving parallel repetition theorems for general games, these works prove general \emph{gap amplification} theorems, which are closely related. Instead of showing that for games $G$ where $\eval(G) < 1$ that $\eval(G^n)$ goes to $0$ with $n$, the game $G$ is first converted to another game $H$ where analyzing $\eval(H^n)$ is much more tractable. Gap amplification is a technique used in complexity theory and cryptography to amplify the difference between two cases of a problem (usually called the \emph{completeness} and \emph{soundness} cases).

Kempe and Vidick showed that given an arbitrary game $G$, one can efficiently transform it to another game $H$ with the following properties: if the \emph{classical} value of $G$ is $1$ (meaning that there is a perfect deterministic strategy), then $\val(H^n) = 1$ (and thus $\eval(H^n) = 1$). If the \emph{entangled} value of $G$ is less than $1$, then the entangled value of $H^n$ decays at a polynomial rate $n^{-\Omega(1)}$. In this tranformed game $H$, in addition to playing the game $G$, the referee will randomly choose to ask ``consistency'' questions to check that the players give the same answers on the same questions\footnote{This transformation is to due to Feige and Kilian~\cite{feige2000two}, who proved a similar result for classical games.}. Thus~\cite{kempe2011parallel} prove gap amplification for general games -- with a caveat. Because of the random consistency checks in the game $H$, the ``quantum completeness'' is not preserved: even if $\eval(G) = 1$, it is not necessarily the case that $\eval(H) = 1$. 

More recently, Bavarian, Vidick, and Yuen~\cite{BVY15anchoring,BVY15fort} gave better gap amplification results for entangled games\footnote{They also obtain general gap amplification results for games with more than two players.}. They showed that for general games $G$, one can apply a simple transformation to obtain another game $H$ with the following properties:
\medskip
\begin{enumerate}
	\item If $\eval(G) = 1$, then $\eval(H^n) = 1$.
	\item If $\eval(G) < 1$, then $\eval(H^n) \leq \exp(-\Omega(n))$.
\end{enumerate}
\medskip
Note that the transformation from $G$ to $H$ preserves quantum completeness, and that when $\eval(G) < 1$, the entangled value of the repeated game decays \emph{exponentially}. Like~\cite{kempe2011parallel}, the transformations of~\cite{BVY15anchoring,BVY15fort} construct $H$ by adding auxiliary questions to the game $G$. The transformation given in~\cite{BVY15anchoring} is called \emph{anchoring}, and the trasformation in~\cite{BVY15fort} is called \emph{fortification}. The latter transformation gives a quantum generalization of the fortification technique of~\cite{moshkovitz2014parallel} for \emph{classical games}. The quantitative aspects of repeated anchored games are different from those of fortified games, but both yield general gap amplification theorems for entangled games.

The results of Bavarian, Vidick and Yuen show that, while we do not know if the Quantum Parallel Repetition Conjecture holds for all games $G$, we \emph{do} know that it holds for a class of games that effectively captures the general case, in fact with exponential decay similar to Raz's theorem. Since the main application of parallel repetition in complexity theory and quantum information is gap amplification, the results of~\cite{BVY15anchoring,BVY15fort} effectively settle the Quantum Parallel Repetition Conjecture -- as far as applications are concerned. 

But as a scientific question, the original Quantum Parallel Repetition Conjecture is a fundamental and basic problem about the power of entanglement in games. Prior to this work, one might have wondered whether there exists a game $G$ such that $\eval(G) < 1$, but there is some constant $\delta$ such that for infinitely many $n$ there is a nefarious entangled strategy for $G^n$ with success probability at least $\delta$? Here we prove that this cannot happen.

\subsection{Proof overview}

Theorem~\ref{thm:main} is proved via reduction: if $\eval(G^n)$ is too large, then from an optimal entangled strategy for $G^n$ we can construct an entangled strategy for the single-shot game $G$ that wins with probability strictly greater than $\eval(G)$, which would be a contradiction. 

In more detail, suppose that $\eval(G) = 1 - \eps$. If the success probability of the players in $G^n$ is dramatically larger than our target bound (which in our case is $\sim n^{-O(1)}$), then we can identify a set of coordinates $C \subseteq [n]$ that is not too large, but has the property that for a uniformly random coordinate $i \in [n] - C$, 
\begin{equation}
	\Pr(\text{Win game $i$ } | \text{Win games in $C$}) > 1 - \eps/2
	\label{eq:embed}
\end{equation}
where here the probability is both over the randomness of the questions in $G^n$, the randomness of the players' entangled strategy, and the randomly chosen index $i$. Thus it would be advantageous if Alice and Bob could play the single-shot game $G$ by ``embedding'' it in a randomly chosen $i$th coordinate of $G^n$, and playing $G^n$ \emph{conditioned} on the event that the games indexed by $C$ have been won. If they could do this, then by~\eqref{eq:embed}, the probability they win the $i$th coordinate of $G^n$, and hence the original game $G$, is at least $1 - \eps/2 > \eval(G)$, which would be a contradiction.

If the players are classical (i.e. use deterministic strategies), this embedding is performed in the following way. Alice and Bob are first given questions $(X_i,Y_i)$ for the $i$'th game. Based on their received question, Alice and Bob jointly sample a \emph{dependency-breaking variable} $R$. The essential features of this dependency-breaking variable are:
\medskip
\begin{enumerate}
	\item \textbf{Usefulness}\footnote{We will let $\P$ denote the probability distribution that describes the joint distribution of the random variables relevant in an execution of the strategy for $G^n$, including the players' questions $X_1,\ldots,X_n,Y_1,\ldots,Y_n$, the players' answers $A_1,\ldots,A_n,B_1,\ldots,B_n$, and the dependency-breaking variable $R$.
}: $\P_{A_i B_i | R X_i Y_i W_C} = \P_{A_i | R X_i W_C} \cdot \P_{B_i | R Y_i W_C}$
	\item \textbf{Sampleability}: $\P_{R | X_i Y_i W} \approx \P_{R | X_i W_C} \approx \P_{R | Y_i W_C}$
\end{enumerate}
\medskip
where  ``$\approx$'' means closeness in statistical distance. Here, $W_C$ denotes the event that the players win all the games in $C$. $\P_{A_i B_i | R X_i Y_i W_C}$ denotes the probability distribution of Alice's and Bob's answers in the $i$th coordinate when playing $G^n$, conditioned on the dependency-breaking variable $R$, their received questions for the $i$th game $(X_i, Y_i)$, \emph{and} the event $W_C$. The ``Usefulness property'' states that, the players' answers in the $i$th round are independent of each other, conditioned on $R$, their own questions, and $W_C$. Thus, given $R$ distributed according to $\P_{R | X_i Y_i W_C}$, Alice can sample $A_i$ on her own, because she possesses $R$ and $X_i$, and similarly Bob can sample $B_i$ on his own, because he possesses knowledge of $R$ and $Y_i$. By~\eqref{eq:embed}, the probability that $V(X_i,Y_i,A_i,B_i) = 1$ will be strictly greater than $\eval(G)$, wherein we would arrive at a contradiction.

As the name suggests, the ``sampleability property'' implies that Alice and Bob can (approximately) jointly sample the variable $R$. Even though the distribution $\P_{R | X_i Y_i W_C}$ may depend on both players' questions, the sampleability property shows $R$, up to some error, only depends on $X_i$ or $Y_i$, but not both. Using the \emph{correlated sampling procedure} of~\cite{Hol09}, Alice and Bob can jointly sample $R$ from $\P_{R | X_i Y_i W}$ with high probability.

At a high level, the proof of our quantum parallel repetition theorem is similar. However instead of sampling a dependency-breaking variable $R$, the players will need to sample a \emph{dependency-breaking state}. It is an entangled state $\ket{\Psi_{x_iy_i}}$ that depends on both Alice's and Bob's questions $(x_i,y_i)$, and satisfies similar Usefulness and Sampleability properties:
\medskip
\begin{enumerate}
	\item \textbf{Usefulness}: The distribution of measurement outcomes by making local measurements on $\ket{\Psi_{X_iY_i}}$ is equal to $\P_{A_i B_i | X_i Y_i W_C}$. 
	\item \textbf{Sampleability}: There exist states $\ket{\Phi_{X_i}}$ and $\ket{\Gamma_{Y_i}}$ such that $\ket{\Psi_{X_iY_i}} \approx \ket{\Phi_{X_i}} \approx \ket{\Gamma_{Y_i}}$
\end{enumerate}
\medskip
where ``$\approx$'' means closeness in $\ell_2$ distance, and the statements hold on average over $X_i Y_i$.

The Usefulness property states that if on input $(x_i,y_i)$, Alice and Bob were to share the entangled state $\ket{\Psi_{x_iy_i}}$, then they could make local measurements to obtain outcomes distributed according to $\P_{A_i B_i | X_i Y_i W_C}$, which would mean that their success probability would be $\Pr(\text{Win $i$ }  | \text{Win $C$})$, which is greater than $\eval(G)$, an impossibility.

The Sampleability property implies that on input $(x_i,y_i)$ Alice and Bob are actually able to approximately prepare the state $\ket{\Psi_{x_iy_i}}$. This is because of the \emph{quantum correlated sampling procedure} of Dinur, Steurer, and Vidick, who used it to prove a parallel repetition theorem for entangled projection games~\cite{DinurSV14}. It is entirely analogous to Holenstein's correlated sampling procedure: Alice has a description of a state $\ket{\Phi_{X_i}}$ that's close to $\ket{\Psi_{X_i Y_i}}$, and Bob has a description of a state $\ket{\Gamma_{Y_i}}$ that is also close to $\ket{\Psi_{X_i Y_i}}$. Via local transformations on preshared quantum entanglement, Alice and Bob can generate an approximation of $\ket{\Psi_{X_i Y_i}}$. Combined with the Usefulness property, Alice and Bob are then able to win the $i$th game with too high probability.

It is not difficult to define states that satisfy the Usefulness property. Consider an execution of the entangled strategy for $G^n$. In the beginning, the players share some entangled state $\ket{\psi}$, and upon obtaining questions $(x_1,\ldots,x_n)$ and $(y_1,\ldots,y_n)$, the players apply local measurements depending on these questions to $\ket{\psi}$ to obtain answer tuples $(a_1,\ldots,a_n)$ and $(b_1,\ldots,b_n)$. One can define an ensemble of states $\{\ket{\Psi_{x_i,y_i}}\}$ that are, roughly speaking, derived from the post-measurement state of the players \emph{conditioned} on the players having won all the games in $C$ (that is, conditioned on the event $W_C$), \emph{and} having received a specific question pair $(x_i,y_i)$ in the $i$'th coordinate. Such an ensemble of states would satisfy the Usefulness property.


However, the primary challenge is achieving Sampleability property, that is, to show the states $\ket{\Psi_{x_i,y_i}}$ only depend on one player's question, but not both. One major obstacle to proving the Sampleability property is the following: in the players' strategy for $G^n$, Bob (say) may elect to ``print'' his entire vector of questions $(y_1,\ldots,y_n)$ into the entangled state $\ket{\psi}$. He can do this by applying a local unitary operation controlled on his questions on some ancilla qubits in $\ket{\psi}$. We cannot say he does not do this, because the shared entangled state $\ket{\psi}$ and the players' measurements are completely arbitrary. But this implies that we cannot hope to prove that the post-measurement state is independent of $y_i$, conditioned on $x_i$. 

Despite such barriers, we are able to define the $\ket{\Psi_{x_i,y_i}}$ in such a way that removes such adversarial dependencies on the players' questions. Assuming (for contradiction) that the players' probability of success is at least $n^{-O(1)}$, then we are able to prove that these states satisfy the Sampleability property. We build upon many previous works: we use the information theoretic framework of~\cite{chailloux2014parallel, JainPY14}, carefully combined with the operator analysis techniques from~\cite{DinurSV14}. The definition of the dependency-breaking states $\ket{\Psi_{x_i,y_i}}$ includes the classical dependency-breaking variables of~\cite{Hol09} used to prove Raz's parallel repetition theorem. Our final constructed strategy for the single-shot game $G$ uses both classical and quantum correlated sampling procedures.

\section{Preliminaries}\label{sec:prelim}
\subsection{Probability distributions}\label{subsec:prob_dist}
We largely adopt the notational conventions from~\cite{Hol09} for probability distributions. We let capital letters denote random variables and lower case letters denote specific samples. We will use subscripted sets to denote tuples, e.g., $X_{[n]} := (X_1,\ldots,X_n)$, $x_{[n]} = (x_1,\ldots,x_n)$, and if $C \subset [n]$ is some subset then $X_C$ will denote the sub-tuple of $X_{[n]}$ indexed by $C$. We use $\P_X$ to denote the probability distribution of random variable $X$, and $\P_X(x)$ to denote the probability that $X = x$ for some value $x$. For multiple random variables, e.g., $X, Y, Z$, $\P_{XYZ}(x,y,z)$  denotes their joint distribution with respect to some probability space understood from context. 

We use $\P_{Y | X = x}(y)$ to denote the conditional distribution $\P_{YX}(y,x)/\P_X(x)$, which is defined when $\P_X(x) > 0$. When conditioning on many variables, we usually use the shorthand $\P_{X | y,z}$ to denote the distribution $\P_{X | Y =y,Z=z}$. For example, we write $\P_{V | \omega_\mi, x_i, y_i}$ to denote $\P_{V | \Omega_\mi = \omega_\mi, X_i = x_i, Y_i = y_i}$. For an event $W$ we let $\P_{X Y | W}$ denote the distribution conditioned on $W$. We use the notation $\Ex_{X} f(x)$ and $\Ex_{\P_X} f(x)$ to denote the expectation $\sum_{x} \P_X(x) f(x)$. 

Let $\P_{X_0}$ be a distribution of $\X$, and for every $x$ in the support of $\P_{X_0}$, let $\P_{Y | X_1 = x}$ be a conditional distribution defined over $\Y$. We define the distribution $\P_{X_0} \P_{Y | X_1}$ over $\X \times \Y$ as
$$
	(\P_{X_0} \P_{Y | X_1})(x,y) \,:=\, \P_{X_0}(x) \cdot \P_{Y | X_1 = x}(y).
$$
Additionally, we write $\P_{X_0 Z} \P_{Y | X_1}$ to denote the distribution $(\P_{X_0 Z} \P_{Y | X_1})(x,z,y) := \P_{X_0 Z}(x,z) \cdot \P_{Y | X_1 = x}(y)$.

For two random variables $X_0$ and $X_1$ over the same set $\X$, we use
$$\| \P_{X_0} - \P_{X_1} \| \,:=\, \frac{1}{2}\sum_{x \in \X} |\P_{X_0}(x) - \P_{X_1} (x)|,$$
to denote the total variation distance between $\P_{X_0}$ and $\P_{X_1}$.

\subsection{Quantum information theory}

For comprehensive references on quantum information we refer the reader to~\cite{nielsen2010quantum,wilde2013quantum}. 

For a vector $\ket{\psi}$, we use $\| \ket{\psi} \|$ to denote its Euclidean length. For a matrix $A$, we will use $\| A \|_1$ to denote its \emph{trace norm} $\Tr(\sqrt{A A^\dagger})$, and $\|A\|_F$ to denote its \emph{Frobenius norm} $\sqrt{\Tr(A A^\dagger)}$. A density matrix is a positive semidefinite matrix with trace $1$. The \emph{fidelity} between two density matrices $\rho$ and $\sigma$ is defined as $F(\rho,\sigma) = \| \sqrt{\rho} \sqrt{\sigma} \|_1$. For Hermitian matrices $A, B$ we write $A \preceq B$ to indicate that $A - B$ is positive semidefinite. We use $\Id$ to denote the identity matrix. A \emph{positive operator valued measurement} (POVM) with outcome set $\A$ is a set of positive semidefinite matrices $\{ E^a \}$ labeled by $a\in \A$ that sum to the identity.

We will use the convention that, when $\ket{\psi}$ is a pure state, $\psi$ refers to the rank-1 density matrix $\ketbra{\psi}{\psi}$. We use subscripts to denote system labels; so $\rho_{AB}$ will denote the density matrix on the systems $A$ and $B$. A \emph{classical-quantum} state (or simply \emph{cq-state}) $\rho_{XE}$ is classical on $X$ and quantum on $E$ if it can be written as $\rho_{XE} = \sum_{x} p(x) \ketbra{x}{x}_X \otimes \rho_{E | X = x}$ for some probability measure $p(\cdot)$. The state $\rho_{E | X = x}$ is by definition the $E$ part of the state $\rho_{XE}$, conditioned on the classical register $X = x$. We write $\rho_{XE | X = x}$ to denote the state $\ketbra{x}{x}_X \otimes \rho_{E | X = x}$. We often write expressions such as $\rho_{E|x}$ as shorthand for $\rho_{E|X = x}$ when it is clear from context which registers are being conditioned on. This will be useful when there are many classical variables to be conditioned on. 

The Fuchs-van de Graaf inequalities relate fidelity and trace norm as
\begin{equation}\label{eq:fuchs-graaf}
1 - F(\rho,\sigma) \leq \frac{1}{2} \| \rho - \sigma \|_1 \leq \sqrt{1 - F(\rho,\sigma)^2}.
\end{equation}
When dealing with pure states, we can tighten the relationship between the trace norm and the Euclidean distance:
\begin{fact}
\label{fact:trace_vs_euclid}
For pure states $\ket{v}$ and $\ket{w}$, $ \left \| \, \ketbra{v}{v} - \ketbra{w}{w} \, \right \|_1 \leq 2 \left \| \,  \ket{v} - \ket{w} \,  \right \|$.
\end{fact}

\noindent \textbf{Ando's Identity.} For any symmetric pure state $\ket{\psi} = \sum_j \sqrt{\lambda_j} \ket{v_j} \ket{v_j}$ for an orthonormal basis $\{\ket{v_j}\}$ and arbitrary linear operators $X, Y$, we have
	$$
		\bra{\psi} X \otimes Y \ket{\psi} = \Tr(X \sqrt{\rho} Y^\top \sqrt{\rho}),
	$$
	where $\rho = \sum \lambda_j \ket{v_j}\bra{v_j}$ is the reduced density matrix of $\ket{\psi}$ on either subsystem and the transpose is taken with respect to the basis $\{\ket{v_j}\}$.

\medskip
\noindent \textbf{Information theoretic quantities.} For two positive semidefinite operators $\rho$, $\sigma$, the \emph{relative entropy} $S(\rho \| \sigma)$ is defined to be $\Tr(\rho (\log \rho - \log \sigma))$. The \emph{relative min-entropy} $S_\infty(\rho \| \sigma)$ is defined as $\min\{ \lambda : \rho \preceq 2^\lambda \sigma \}$.   

Let $\rho_{AB}$ be a bipartite state. The mutual information $I(A:B)_\rho$ is defined as $S(\rho_{AB} \| \rho_A \otimes \rho_B)$. For a classical-quantum state $\rho_{XAB}$ that is classical on $X$ and quantum on $AB$, we write $I(A ; B | x)_{\rho}$ to indicate $I(A ; B)_{\rho_x}$. 

\begin{fact}
\label{fact:rel_inf}
For all states $\rho_{AB}$, $\sigma_A$, and $\tau_B$, we have
$$
	S(\rho_{AB} \| \sigma_A \otimes \tau_B) \geq S(\rho_{AB} \| \rho_A \otimes \rho_B) = I(A ; B)_\rho.
$$
\end{fact}

\begin{fact}[Pinsker's inequality]
\label{fact:pinsker}
	For all density matrices $\rho, \sigma$, $ \frac{1}{2} \| \rho - \sigma \|^2_1 \leq S(\rho \| \sigma)$.
\end{fact}

\begin{lemma}[\cite{JainPY14}, Fact II.8]
\label{lem:divergence_chain_rule}
	Let $\rho = \sum_z \P_Z(z) \ketbra{z}{z} \otimes \rho_z$, and $\rho' = \sum_z \P_{Z'}(z) \ketbra{z}{z} \otimes \rho'_z$. Then $S(\rho' \| \rho) = S(\P_{Z'} \| \P_Z) + \Ex_{Z'} \left [ S(\rho'_z \| \rho_z) \right]$. In particular, $S(\rho' \| \rho) \geq \Ex_{Z'} \left [ S(\rho'_z \| \rho_z) \right]$.
\end{lemma}
We will also use the following Lemma from \cite{chung2015parallel,BVY15anchoring}.

\begin{lemma}[\cite{chung2015parallel,BVY15anchoring}, Quantum Raz's Lemma] \label{lem:quantum_raz} 
Let $\rho$ and $\sigma$ be two cq-states with  $\rho_{XA}= \rho_{X_1 X_2 \ldots X_n A}$ and $\sigma= \sigma_{XA}= \sigma_{X_1}\otimes \sigma_{X_2}\otimes \ldots \otimes \sigma_{X_n} \otimes \sigma_A$ with $X=X_1 X_2 \ldots X_n$ classical in both states. Then
\begin{equation}\label{eqn:Raz_lemma1} \sum_{i=1}^n I(X_i \, :\, A)_\rho \leq S(\rho_{XA} \, \| \sigma_{XA}). \end{equation}
\end{lemma}

\subsection{Classical and quantum correlated sampling}

\emph{Correlated sampling} is a key component of Holenstein's proof of the classical parallel repetition theorem.

\begin{lemma}[Classical correlated sampling~\cite{Hol09}]
\label{lem:ccorsamp}
	Let $\P$ and $\Q$ be two probability distributions over a universe $\mathcal{U}$ such that $\| \P - \Q \|_1 \leq \eps < 1$. Then there exists a zero communication two-player protocol using shared randomness where the first player outputs an element $p \in \mathcal{U}$ distributed according to $\P$, the second player samples an element $q \in \mathcal{U}$ distributed according to $\Q$, and with probability at least $1 - O(\eps)$, the two elements are identical (i.e. $p = q$).
\end{lemma}
We call the protocol in the Lemma above the \emph{classical correlated sampling procedure}. The next lemma is the quantum extension of the correlated sampling lemma, proved by~\cite{DinurSV14} in order to obtain a parallel repetition theorem for entangled projection games, a class of two-player games. Their lemma is a robust version of the quantum state embezzlement procedure of~\cite{van2003universal}. 

\begin{lemma}[Quantum correlated sampling~\cite{DinurSV14}]
\label{lem:qcorsamp}
	Let $d$ be an integer and $\alpha > 0$. Then there exists an integer $d'$ depending on $d$ and $\alpha$, and a collection of unitaries $V_\psi$, $W_\psi$ acting on $\C^{d d'}$ for every state $\ket{\psi} \in \C^d \otimes \C^d$, such that the following holds: for any two states $\ket{\varphi}, \ket{\theta} \in \C^d \otimes \C^d$, 
	$$
		\| \overline{V}_\varphi \otimes W_\theta \ket{E_{dd'}} - \ket{\varphi} \ket{E_{d'}} \| \leq O(\max \{ \alpha^{1/12}, \left \| \, \ket{\varphi} - \ket{\theta} \, \right \|^{1/6} \} )
	$$
	where $\ket{E_{d}} \propto \sum_{j=1}^{d} \frac{1}{\sqrt{j}} \ket{j}\ket{j} $ is the $d$-dimensional embezzlement state.
\end{lemma}
We shall call the protocol in the Lemma above the \emph{quantum correlated sampling procedure}.

\section{Proof of the Main Theorem}

Let $G$ be a two-player one-round game with question distribution $\mu$ and referee predicate $V(x,y,a,b)$. Let $\A$ and $\B$ denote the alphabets of Alice's and Bob's answers, respectively. Let $\eval(G) = 1 - \eps$.

Consider an optimal entangled strategy for $G^n$, which consists of a shared entangled state $\ket{\psi}^{E_A E_B} \in \C^d \otimes \C^d$ and measurement POVMs for Alice and Bob, $\{A_{\xvec}^{\avec} \}$ and $\{B_{\yvec}^{\bvec}\}$ respectively. We will assume that $\ket{\psi}$ is symmetric; i.e., $\ket{\psi} = \sum_i \sqrt{\lambda_i} \ket{v_i} \ket{v_i}$ for some orthonormal basis $\{ \ket{v_i} \}$. This is without loss of generality, as we can always rotate (say) Bob's basis vectors to match Alice's basis vectors, and fold the unitary rotation into Bob's measurements. For $i \in [n]$, let $W_i$ denote the event that the players win coordinate $i$ using this optimal strategy. Let $W = W_1 \wedge \cdots \wedge W_n$ denote the event that the players win all coordinates. For a set $C \subseteq [n]$, let $W_C = \wedge_{i \in C} W_i$.

\begin{proposition}
\label{prop:subset}
	Suppose that $\log 1/\Pr(W) \leq \eps n/16 - \log 4/\eps$. Then there exists a set $C \subseteq [n]$ of size at most $t = \frac{8}{\eps} \left( \log 4/\eps + \log 1/\Pr(W) \right)$ such that
	$$
		\Pr_{i \notin C} (W_i | W_C) \geq 1 - \eps/2.
	$$
	where $i$ is chosen uniformly from $[n] - C$.
\end{proposition}
\begin{proof}
	Set $\delta = \eps/8$. Let $W_{> 1 - \delta}$ denote the event that the players won more than $(1 - \delta)n$ rounds. To show existence of such a set $C$, we will show that $\Ex_C \Pr(\neg W_i | W_C) \leq \eps/2$, where $C$ is a (multi)set of $t$ independently chosen indices in $[n]$. This implies that there exists a particular set $C$ such that $\Pr(\neg W_i | W_C) \leq \eps/2$, which concludes the claim.
	
	First we write, for a fixed $C$,
	\begin{align*}
		\Pr ( \neg W_i | W_C) &= \Pr(\neg W_i | W_C, W_{> 1 - \delta}) \Pr(W_{> 1 - \delta} | W_C) + \\ &\qquad \qquad \Pr(\neg W_i | W_C, \neg W_{> 1 - \delta}) \Pr(\neg W_{> 1 - \delta} | W_C).
	\end{align*}
	Observe that $\Pr(\neg W_i | W_C \wedge W_{> 1 - \delta})$ is the probability that, conditioned on winning all rounds in $C$, the randomly selected coordinate $i \in [n] - C$ happens to be one of the (at most) $\delta n$ lost rounds. This is at most $\delta n/(n - t) \leq \eps/4$, where we use our assumption on $t$ from the Proposition statement. Now observe that 
	\begin{align*}
		\Ex_C \Pr(\neg W_{> 1 - \delta} | W_C) &\leq \Ex_C \frac{\Pr(W_C | \neg W_{> 1 - \delta})}{\Pr(W_C)} \\
								   &\leq \frac{1}{\Pr(W)} (1 - \delta)^t \\
								   &\leq \eps/4
	\end{align*}
	where in the second line we used the fact that $\Pr(W_C) \geq \Pr(W)$.	
\end{proof}

For the rest of the proof we will fix a set $C$ given by Proposition~\ref{prop:subset}. 

\subsection{Dependency-breaking variables}
\label{sec:classical_setup}

We introduce the random variables that play an important role in the proof of Theorem~\ref{thm:main}. Let $C \subseteq [n]$ be as given by Proposition~\ref{prop:subset}. We fix $C=\{m+1, m+2, \ldots, n\}$, where $m=n-|C|$, as this will easily be seen to hold without loss of generality. Let $(X_{[n]},Y_{[n]})$ be distributed according to $\mu_{[n]}$ and $(A_{[n]},B_{[n]})$ be defined from  $X_{[n]}$ and $Y_{[n]}$ as follows: 
$$
	\P_{A_{[n]} B_{[n]}|x_{[n]},y_{[n]}}(a_{[n]},b_{[n]}) = \bra{\psi} A_\xvec^\avec \otimes B_\yvec^\bvec \ket{\psi}.
$$
Let $(X_C,Y_C)$ and $\ac=(A_C,B_C)$ be random variables that denote the players' questions and answers respectively associated with the coordinates indexed by $C$. 

We use the random variables $\Omega$ and $R$ that are crucially used in Holenstein's proof of Raz's parallel repetition theorem. Let $D_1,\ldots,D_m$ be independent and uniformly distributed in $\{Alice,Bob\}$. Let $M_1,\ldots,M_m$ be independent random variables defined in the following way: for each $i \in [m]$,
\begin{align*}
	M_i = \left \{ \begin{array}{ll}
		X_i & \mbox{ if } D_i = Alice \\
		Y_i & \mbox{ if } D_i = Bob
		\end{array}
	\right.
\end{align*}
Now for $i \in [m]$, we define $\Omega_i := (D_i,M_i)$. We say that $\Omega_i$ \emph{fixes Alice's input} if $D_i = Alice$, and otherwise $\Omega_i$ fixes Bob's input. We write $\Omega$ to denote the random variable $(\Omega_1,\ldots,\Omega_m,X_C,Y_C)$, where $X_CY_C$ are Alice and Bob's questions in the coordinates indexed by $C$. For $i \in [m]$ we write $\Omega_\mi$ to denote the random variable $\Omega$ with $\Omega_i$ omitted. 

\begin{proposition}
	Conditioned on $\Omega$, $X_{[n]}$ and $Y_{[n]}$ are independent.
\end{proposition}

Finally, we will define a \emph{dependency-breaking variable} $R := (\Omega,A_C,B_C)$, where $A_C$ and $B_C$ are the players' answers in the coordinates indexed by $C$. For $i \notin C$, we let $R_\mi := (\Omega_\mi,A_C,B_C)$. $R_i$ will refer to $\Omega_i$. We will use lowercase letters to denote instantiations of these random variables: e.g., $r_\mi$, $x_i$, and $y_i$ refer to specific values of $R_\mi$, $X_i$, and $Y_i$.

Throughout our proofs, all expectations are implicitly over the measure defined by $\P$. For example, the expectation $\Ex_{\Omega_\mi Z | x_i, y_i}$ indicates $\sum_{\omega_\mi, a_C, b_C} \P_{\Omega_\mi A_C B_C | x_i, y_i} (\omega_\mi, a_C, b_C)$. Given an event such as $W$ (winning all the coordinates) or $W_C$ (winning all the coordinates in $C$), $\P(W)$ and $\P(W_C)$ will mean the probability of these events with respect to the distribution $\P$. 

The following Lemma expresses the idea that, because $W_C$ is an event that occurs with not-too-small probability, conditioning on it cannot skew the distribution of variables corresponding to an average coordinate by too much. This Lemma follows in a straightforward manner from the~\cite{Hol09}.
\begin{lemma}
\label{lem:classical_skew}
The following statements hold on, average over $i$ chosen uniformly in $[m]$:
\medskip
\begin{enumerate}
	\item $\Ex_i \| \P_{R_i X_i Y_i | W_C} - \P_{R_i X_i Y_i} \|_1 \leq O(\sqrt{\delta})$
	\item $\Ex_i \left \| \P_{X_i Y_i R_\mi | W_C} - \P_{X_i Y_i} \cdot \P_{R_\mi | X_i W_C} \right \|_1 \leq O(\sqrt{\delta})$
		\item $\Ex_i \left \| \P_{X_i Y_i R_\mi | W_C} - \P_{X_i Y_i} \cdot \P_{R_\mi | Y_i W_C} \right \|_1 \leq O(\sqrt{\delta})$
\end{enumerate}
\medskip
where $\delta := \frac{1}{m} \left ( \log 1/\P(W_C) + |C| \log |\A| |\B| \right)$.
\end{lemma}

\subsection{Two key Lemmas, and proof of the Main Theorem}

For every $i \in [n] - C$, we will construct a collection of bipartite states $\{ \ket{\Psi_{r_\mi,x_i,y_i}} \} \subseteq \C^d \otimes \C^d$, which we call dependency-breaking states, that are indexed by the dependency-breaking variable $r_\mi$ defined above, and questions $(x_i,y_i)$. The following lemmas state the important properties of this collection of states:

\begin{lemma}[Usefulness Lemma]
\label{lem:useful}
	For all $r_\mi, x_i, y_i$, there exist POVMs $\{ \what{A}_{r_\mi,x_i}^{a_i} \}$ and $\{\what{B}_{r_\mi,y_i}^{b_i} \}$ acting on $\C^d$ such that
	$$
		\P_{A_i B_i | r_\mi, x_i, y_i} (a_i,b_i) = \Tr \left (  \what{A}_{r_\mi,x_i}^{a_i} \otimes \what{B}_{r_\mi,y_i}^{b_i} \, \Psi_{r_\mi,x_i,y_i} \right).
	$$
\end{lemma}

\begin{lemma}[Sampleability Lemma]
\label{lem:sampleable}
There exists an integer $d' \geq d$ such that for every $i, r_\mi, x_i, y_i$, there exist local unitaries $U_{r_\mi,x_i}, V_{r_\mi,y_i}$ acting on $\C^{d'}$ such that
$$
	\Ex_i \Ex_{X_i Y_i} \left [ \Ex_{R_\mi | x_i,y_i, W_C} \left \| U_{r_\mi,x_i} \otimes V_{r_\mi,y_i} \ket{E_{dd'}} - \ket{\Psi_{r_\mi,x_i,y_i}} \ket{E_{d'}} \right \| \right ] \leq O( ( \delta^{1/4} / \P(W_C))^{1/12} )
$$
where $\ket{E_{dd'}}$ and $\ket{E_{d'}}$ are $dd'$ and $d'$-dimensional embezzlement states, respectively, and $\delta$ is defined to be $\frac{1}{m} \left ( \log 1/\P(W_C) + |C| \log |\A| |\B| \right)$.
\end{lemma}

Lemma~\ref{lem:useful} shows that the states $\ket{\Psi_{r_\mi,x_i,y_i}}$ are \emph{useful} to have; they allow Alice and Bob to produce answers in the $i$'th coordinate whose statistics are consistent with the dependency-breaking variable $r_\mi$ and their inputs $(x_i,y_i)$. Lemma~\ref{lem:sampleable} shows that these states are \emph{locally generatable} by Alice and Bob, when given joint access to preshared entanglement, the dependency-breaking variable $r_\mi$ and their own inputs $x_i$ and $y_i$ respectively.

Using these two Lemmas we can prove the Main Theorem.

\begin{proof}[Proof of the Main Theorem]
Consider the following strategy for the game $G$. Alice and Bob share beforehand the embezzlement state $\ket{E_{dd'}}$ of dimension $dd'$ given by Lemma~\ref{lem:sampleable}, and they also have access to shared randomness. Given inputs $(x_i,y_i)$ distributed according to $\P_{X_i Y_i} = \mu$:
\medskip
	\begin{enumerate}
		\item Alice and Bob jointly sample a uniformly random $i \in [n] - C$.
		\item Alice and Bob jointly, approximately sample $R_\mi$ from $\P_{R_\mi | x_i, y_i, W_C}$ using the classical correlated sampling procedure.
		\item Alice applies $U_{r_\mi,x_i}$ to her side of $\ket{E_{dd'}}$
		\item Bob applies $V_{r_\mi,y_i}$ to his side of $\ket{E_{dd'}}$
		\item Alice measures her side of the entanglement using $\{ \what{A}_{r_\mi,x_i}^{a_i} \}$ and outputs the outcome $a_i$
		\item Bob measures his side of the entanglement using $\{ \what{B}_{r_\mi,y_i}^{b_i} \}$ and outputs the outcome $b_i$
	\end{enumerate}
	\medskip
	We now analyze the success probability of this strategy. We will use $\wt{\P}$ to denote the distribution of variables in the probability space associated with an execution of this strategy. For example, we will write $\wt{\P}_{R_\mi | X_i Y_i}$ to denote the distribution of $R_\mi$ conditioned on $X_i Y_i$ that is sampled in Step 1. From Lemma~\ref{lem:classical_skew} we have that on average over $i$, $\P_{X_i Y_i R_\mi  | W_C} \approx \P_{X_i Y_i} \cdot \P_{R_\mi | X_i W_C} \approx \P_{X_i Y_i} \cdot \P_{R_\mi | Y_i W_C}$, where ``$\approx$'' means closeness in statistical distance. By invoking the classical correlated sampling procedure of Lemma~\ref{lem:ccorsamp}, we get
	$$
		\Ex_i \| \P_{X_i Y_i} \cdot \wt{\P}_{R_\mi | X_i Y_i} - \P_{X_i Y_i R_\mi | W_C} \|_1 \leq O(\sqrt{\delta}).
	$$
	After Step 3, Alice and Bob will possess a state $\ket{\Lambda_{r_\mi,x_i,y_i}}$ such that
	$$
		\Ex_i \Ex_{X_i Y_i} \left [ \Ex_{R_\mi | x_i, y_i, W_C} \| \Lambda_{r_\mi,x_i,y_i} - \Psi_{r_\mi,x_i,y_i} \|_1 \right] \leq \eta
	$$
	where $\eta = O( ( \delta^{1/4} / \P(W_C))^{1/12} )$. Consider the measurement process in Steps 4 and 5. Let $\wt{\P}_{A_i B_i | r_\mi, x_i, y_i}$ denote the distribution of measurement outcomes in this strategy, conditioned on their inputs and a sampled value of $r_\mi$.  By Lemma~\ref{lem:useful} and the fact that the trace norm is nonincreasing under quantum operations, we have that
	$$
		\Ex_i \Ex_{X_i Y_i} \left [ \Ex_{R_\mi | x_i, y_i, W_C} \| \wt{\P}_{A_i B_i | x_i, y_i, r_\mi} - \P_{A_i B_i | x_i, y_i, r_\mi} \|_1 \right ] \leq \eta
	$$
	or equivalently 
	$$
		\Ex_i \| \P_{X_i Y_i} \cdot \wt{\P}_{R_\mi | X_i Y_i W_C} \cdot \wt{\P}_{A_i B_i | x_i, y_i, r_\mi} - \P_{X_i Y_i} \cdot \P_{R_\mi | X_i Y_i W_C} \cdot\P_{A_i B_i R_\mi | X_i Y_i W_C} \|_1 \leq \eta.
	$$
	By Lemma~\ref{lem:classical_skew} we have $\Ex_i \| \P_{X_i Y_i | W_C} - \P_{X_i Y_i} \| \leq \sqrt{\delta}$. By triangle inequality and that $\wt{\P}_{X_i Y_i} = \P_{X_i Y_i}$, we have
	$$
		\Ex_i \| \wt{\P}_{X_i Y_i R_\mi A_i B_i} - \P_{X_i Y_i R_\mi A_i B_i | W_C} \|_1 \leq O(\eta).
	$$
	Note that $\wt{\P}_{X_i Y_i R_\mi A_i B_i}$ represents the probability distribution of all the variables present in the strategy above. Let $W_i$ denote the probability the players win the $i$th coordinate. Thus we get 
	\begin{align}
		\Ex_i | \wt{\P}(W_i) - \P(W_i | W_C) | \leq O(\eta).
		\label{eq:main_bound}
	\end{align}
	Assume that 
	$$
		\P(W) \geq \frac{c s \log n}{\eps^{17} n^{1/4}}
	$$
	where $c > 0$ is a universal constant, and $s$ is the bit-length of the players' answers. Since $\P(W_C) \geq \P(W)$, and using our bound on $|C|$ (from Proposition~\ref{prop:subset}) and our bound on $\delta$ (from Lemma~\ref{lem:classical_skew}), this implies that the right hand side of~\eqref{eq:main_bound} is at most $\eps/4$ (for an appropriate choice of $c$). This implies that 
	\begin{align*}
		\Ex_i \wt{\P}(W_i) &\geq \Ex_i \P(W_i | W_C) - \eps/4 \\
						&\geq 1 - \eps/2 - \eps/4 \\
						&> \eval(G)
	\end{align*}
	where in the second line we used the bound from Proposition~\ref{prop:subset}. However, this implies that there exists an $i$ such that $\wt{\P}(W_i) > \eval(G)$, which is a contradiction. Therefore $\P(W) \leq \frac{c s \log n}{\eps^{17} n^{1/4}}$.

\end{proof}


\section{Proofs of the two Key Lemmas}

Now we turn to proving the two key lemmas above, the Usefulness Lemma and the Sampleability Lemma.

\subsection{Quantum states and operators}
\label{sec:quantum_setup}

In this subsection we define the states $\ket{\Psi_{r_\mi,x_i,y_i}}$ and measurement operators $\{ \what{A}_{r_\mi,x_i}^{a_i} \}$ and $\{\what{B}_{r_\mi,y_i}^{b_i} \}$. Recall that the dependency-breaking variable $R$ consists of the set of fixed questions $\Omega = (X_C,Y_C,\Omega_1,\ldots,\Omega_m)$ and fixed answers $Z = (A_C,B_C)$ for the coordinates in $C$. 

\medskip
\noindent \textbf{Coarse-grained measurements.} We first \emph{coarsen} the measurement POVMs $\{A_{\xvec}^{\avec} \}$ and $\{B_{\yvec}^{\bvec}\}$ that constitute Alice and Bob's strategy in $G^n$ to construct a set of \emph{intermediate measurements}, which essentially produce answers for the games in set $C$, conditioned on a setting of $\Omega$.

Fix $i$, $\omega$, $a_C$, $b_C$, $x_i$, $y_i$. 
Define
 \begin{align*}
 A_{\omega_\mi,x_i}^{a_C} = \sum_{\avec | a_C} \Ex_{\Xvec | \omega_\mi, x_i} A_\xvec^{\avec} \qquad \qquad  \qquad B_{\omega_\mi,y_i}^{b_C} = \sum_{\bvec | b_C} \Ex_{\Yvec | \omega_\mi, y_i} B_\yvec^{\bvec}
\end{align*}
where $\avec | a_C$ (resp. $\bvec | b_C$) indicates summing over all tuples $\avec$ consistent with the suffix $a_C$ (resp. $\bvec$ consistent with suffix $b_C$) and recall that $\Ex_{\Xvec | \omega_\mi, x_i}$ is shorthand for $\sum_{x_{[n]}} \P_{\Xvec | \Omega_\mi = \omega_\mi, X_i = x_i}(x_{[n]})$. 
We also define
 \begin{align*}
 A_{\omega}^{a_C} = \Ex_{\Xvec | \omega} A_\xvec^{a_C} \qquad \qquad \qquad  B_{\omega}^{b_C} = \Ex_{\Yvec | \omega} B_\yvec^{b_C}.
\end{align*}

Let $\rho$ denote the reduced density matrix of $\ket{\psi}$ on Alice's side. Since we have assumed that $\ket{\psi}$ is symmetric, $\rho$ is also the reduced density matrix on Bob's side. For all $i$, $\omega$, $x_i,y_i,a_C,b_C$, let $U_{\omega_\mi, x_i, a_C}$, $U_{\omega,a_C}$, $V_{\omega_\mi, y_i, b_C}$, and $V_{\omega,b_C}$ be unitaries such that
\begin{align*}
&U_{\omega_\mi, x_i, a_C} (A_{\omega_\mi, x_i}^{a_C})^{1/2} \sqrt{\rho} \qquad &V_{\omega_\mi, y_i, b_C} (B_{\omega_\mi, y_i}^{b_C})^{1/2} \sqrt{\rho} \\
&U_{\omega, a_C} (A_{\omega}^{a_C})^{1/2} \sqrt{\rho} \qquad &V_{\omega, b_C} (B_{\omega}^{b_C})^{1/2} \sqrt{\rho}
\end{align*}
are positive semidefinite. Such unitaries can be found via singular value decompositions. For notational convenience, let
\begin{align*}
	&S_{\omega_\mi, x_i, a_C} = U_{\omega_\mi, x_i, a_C} (A_{\omega_\mi, x_i}^{a_C})^{1/2} \qquad 
	&T_{\omega_\mi, y_i, b_C} = V_{\omega_\mi, y_i, b_C} (B_{\omega_\mi, y_i}^{b_C})^{1/2} \\
	&S_{\omega, a_C} = U_{\omega, a_C} (A_{\omega}^{a_C})^{1/2} \qquad &T_{\omega, b_C} = V_{\omega, b_C} (B_{\omega}^{b_C})^{1/2}
\end{align*}

\medskip
\noindent \textbf{Fine-grained measurements.} Now we can define the \emph{fine-grained measurements} that Alice and Bob can apply to obtain answers for the $i$'th game. Define
$$
	\what{A}_{r_\mi,x_i}^{a_i} = S_{\omega_\mi,x_i,a_C}^{-1} A_{\omega_\mi,x_i}^{a_C, a_i} S_{\omega_\mi,x_i,a_C}^{-1}
	\qquad \qquad \qquad 
		\what{B}_{r_\mi,y_i}^{b_i} = T_{\omega_\mi,y_i,b_C}^{-1}  B_{\omega_\mi,y_i}^{b_C,b_i}  T_{\omega_\mi,y_i,b_C}^{-1}
$$
where 
$$
A_{\omega_\mi,x_i}^{a_C,a_i} = \sum_{\avec | a_C, a_i} \Ex_{\Xvec | \omega_\mi, x_i} A_{\xvec}^\avec  \qquad \qquad \qquad B_{\omega_\mi,y_i}^{b_C,b_i} = \sum_{\bvec | b_C, b_i} \Ex_{\Yvec | \omega_\mi, y_i} B_{\yvec}^\bvec
$$
and $\avec | a_C, a_i$ (resp. $\bvec | b_C, b_i$) denotes summing over all $\avec$ consistent with $a_C$ and $a_i$ (resp. all $\bvec$ consistent with $b_C$ and $b_i$). It is easy to verify that the sets $\{\what{A}_{r_\mi,x_i}^{a_i}\}_{a_i \in \A}$ and $\{\what{B}_{r_\mi,y_i}^{b_i}\}_{b_i \in \B}$ form POVMs. Here, for a square matrix $A$, $A^{-1}$ denotes its generalized inverse.

\medskip
\noindent \textbf{States.} Now we are ready to define the states. Fix $i$, $r_\mi = (\omega_\mi,a_C,b_C)$, and $x_i, y_i$. Then let
$$
	\ket{\Psi_{r_\mi,x_i,y_i}} = \frac{S_{\omega_\mi,a_C,x_i} \otimes T_{\omega_\mi,b_C,y_i} \ket{\psi}}{ \left \| S_{\omega_\mi,a_C,x_i} \otimes T_{\omega_\mi,b_C,y_i} \ket{\psi} \right \|}.
$$
Observe that the normalization $\left \| S_{\omega_\mi,a_C,x_i} \otimes T_{\omega_\mi,b_C,y_i} \ket{\psi} \right \|^2$ is equal to $\P_{A_C B_C | \omega_\mi, x_i, y_i}(a_C,b_C)$. 

\subsection{Proof of Usefulness Lemma (Lemma~\ref{lem:useful})}
This Lemma follows from a simple calculation: for every $x_i, y_i, a_i, b_i, r_\mi$:
\begin{align*}
&\Tr \left (  \what{A}_{r_\mi,x_i}^{a_i} \otimes \what{B}_{r_\mi,y_i}^{b_i} \, \Psi_{r_\mi,x_i,y_i} \right) \\
&= \frac{1}{\left \| S_{\omega_\mi,a_C,x_i} \otimes T_{\omega_\mi,b_C,y_i} \ket{\psi} \right \|^2} \Tr \left ( A_{\omega_\mi,x_i}^{a_C, a_i} \otimes B_{\omega_\mi,y_i}^{b_C,b_i} \ketbra{\psi}{\psi} \right ) \\
&= \frac{1}{\P_{A_C B_C | \omega_\mi, x_i, y_i}(a_C,b_C)} \sum_{\avec | a_C, a_i} \sum_{\bvec | b_C, b_i} \Ex_{\Xvec \Yvec | \omega_\mi, x_i, y_i} \Tr \left ( A_\xvec^\avec \otimes B_\yvec^\bvec \ketbra{\psi}{\psi} \right) \\
&= \frac{\P_{A_i B_i A_C B_C | \omega_\mi, x_i, y_i}(a_i,b_i,a_C,b_C)}{\P_{A_C B_C | \omega_\mi, x_i, y_i}(a_C,b_C)}  \\
&= \P_{A_i B_i | r_\mi, x_i, y_i} (a_i,b_i).
\end{align*}
In the second equality we used that conditioned on $\Omega$, $X_{[n]}$ and $Y_{[n]}$ are independent, so therefore $\Ex_{X_{[n]} | \omega_\mi, x_i} \Ex_{Y_{[n]} | \omega_\mi,y_i} = \Ex_{\Xvec \Yvec | \omega_\mi, x_i, y_i}$. In the last equality we used that $r_\mi = (\omega_\mi,a_C,b_C)$. This concludes the Usefulness Lemma.

\subsection{Proof of the Sampleability Lemma (Lemma~\ref{lem:sampleable})}

\paragraph{Overview.} Here we give some intuition. We first analyze an ensemble of states $\{ \ket{\Gamma_{x_i,x_C,a_C}} \}$ (for now we omit mention of the dependency-breaking variable $R$ for simplicity). These are indexed by Alice's questions in the $i$'th coordinate, her questions in the $C$ coordinates, as well as her answers in the $C$ coordinates. The state $\ket{\Gamma_{x_i,x_C,a_C}}$ roughly represents the state of the players where only Alice has applied her measurements -- Bob hasn't done anything yet. 

Fix a $y_i$, $x_C$, $a_C$. For average $x_i, x_i'$ that are independently sampled from the marginal distribution $\P_{X_i | Y_i = y_i}$, we will show that 
$$
	\| \ket{\Gamma_{x_i,x_C,a_C}} - \ket{\Gamma_{x_i',x_C,a_C}} \| \sim \frac{1}{n}.
$$
To handle issues such as Alice ``printing'' her input onto the state $\ket{\psi}$ (as discussed in the introduction), the definition of $\ket{\Gamma_{x_i,x_C,a_C}}$ requires local unitaries that ``undo'' such overt actions of Alice and Bob -- this is accomplished by the unitaries $U$ and $V$ defined in Section~\ref{sec:quantum_setup}. 

Then, we consider what happens when we apply Bob's measurement to both states $\ket{\Gamma_{x_i,x_C,a_C}}$ and $\ket{\Gamma_{x_i',x_C,a_C}}$, and condition on obtaining answers $b_C$ for the $C$ coordinates. His measurement will depend on the questions $y_i$ and $y_C$. The post-measurement states will be precisely $\ket{\Psi_{x_i,y_i,x_C,y_C,a_C,b_C}}$ and $\ket{\Psi_{x_i',y_i,x_C,y_C,a_C,b_C}}$. The distance between these states will be, roughly speaking, the distance between $\ket{\Gamma_{x_i,x_C,a_C}}$ and $\ket{\Gamma_{x_i',x_C,a_C}}$ divided by the probability of Bob obtaining outcome $b_C$ conditioned on Alice obtaining $a_C$. If we average this distance over all choices of $x_C, y_C, a_C, b_C$ that imply the event $W_C$, we get that the average distance between $\ket{\Psi_{x_i,y_i,x_C,y_C,a_C,b_C}}$ and $\ket{\Psi_{x_i',y_i,x_C,y_C,a_C,b_C}}$ is approximately $\frac{1}{n \P(W_C)}$. If $\P(W)$ is much greater than $1/n$, then this distance is small. We then invoke quantum correlated sampling (Lemma~\ref{lem:qcorsamp}), and that proves the Sampleability Lemma.

\paragraph{Proof.} We introduce the following state: 
$$
\xi_{\Omega \Xvec E_AE_B A_C} = \sum_{\omega, \xvec, a_C} \P_{\Omega \Xvec} (\omega, \xvec) \, \ketbra{\omega\, \xvec}{\omega\, \xvec}   \otimes \sqrt{A_{\xvec}^{a_C}} \ketbra{\psi}{\psi} \sqrt{A_{\xvec}^{a_C}} \otimes \ketbra{a_C}{a_C}.
$$
If we trace out the $E_A$ register, we have that
	\begin{align}
		\xi_{\Omega \Xvec E_B A_C} &= \sum_{\omega, \xvec, a_C} \P_{\Omega \Xvec} (\omega, \xvec) \, \ketbra{\omega\, \xvec}{\omega\, \xvec}   \otimes \sqrt{\rho} \overline {A_{\xvec}^{a_C}} \sqrt{\rho} \otimes \ketbra{a_C}{a_C} \notag \\
		&\preceq  \sum_{\omega, \xvec, a_C} \P_{\Omega \Xvec} (\omega, \xvec) \, \ketbra{\omega\, \xvec}{\omega\, \xvec}   \otimes \sqrt{\rho} \overline {A_{\xvec}^{a_C}} \sqrt{\rho} \otimes \Id \notag\\
		&= \sum_{\omega, \xvec} \P_{\Omega \Xvec} (\omega, \xvec) \, \ketbra{\omega\, \xvec}{\omega\, \xvec}   \otimes \rho \otimes \Id \notag,
	\end{align}
where $\rho$ is the reduced density matrix of $\ket{\psi} = \sum_j \sqrt{\lambda_j} \ket{v_j}\ket{v_j}$ on $E_B$, $\overline{A_{\xvec}^{a_C}}$ denotes the entry-wise complex conjugate of $A_{\xvec}^{a_C}$ with respect to the basis $\{\ket{v_j}\}$, and the last equality uses $\sum_{a_C} \overline{A_{\xvec}^{a_C}} = \Id$. From the definition of $S_\infty$ we have
	\begin{align*}
	|C| \cdot \log |\A| &\geq S_\infty \left (\xi_{\Omega \Xvec E_B A_C} \, \middle \| \, \xi_{\Omega \Xvec} \otimes \xi_{E_B} \otimes \frac{\Id}{\Tr(\Id)} \right ) \\
	&\geq S \left (\xi_{\Omega \Xvec E_B A_C} \, \middle \| \, \xi_{\Omega \Xvec} \otimes \xi_{E_B} \otimes \frac{\Id}{\Tr(\Id)} \right ) \qquad \qquad &\text{($S(\cdot \| \cdot) \leq S_\infty(\cdot \| \cdot)$)} \\ 
		&\geq \Ex_{\Omega, A_C} S \left (\xi_{\Xvec E_B | \omega, a_C} \, \middle \| \, \xi_{\Xvec | \omega} \otimes \xi_{E_B} \right ) \qquad \qquad &\text{(Lemma~\ref{lem:divergence_chain_rule})} 
\end{align*}
	Now we apply Quantum Raz's Lemma:
\begin{align}
	\Ex_{\Omega, A_C} \Ex_i \,\, I(X_i ; E_B | \omega, a_C)_\xi \leq \frac{|C| \cdot \log |\A|}{m} \leq \delta
	\label{eq:claim23-3}
\end{align}
where recall that we defined $\delta = (|C| \log |\A| \cdot |B|)/m$. Applying the inequalities of Pinsker and Jensen, we obtain
\begin{align}
	\Ex_{\Omega, A_C} \Ex_i \Ex_{X_i | \omega, a_C} \left \|\xi_{E_B | \omega, x_i, a_C} - \xi_{E_B | \omega, a_C} \right \|_1  \leq \sqrt{\delta}.
\label{eq:trace_bound0}
\end{align}
These marginal density matrices have a nice description. Fix $i, \omega, x_i, a_C$. First we note that the state $\xi_{E_B | \omega, x_i, a_C}$ does not depend on $\omega_i$, because we are already conditioning on $x_i$. Thus we can write it as $\xi_{E_B | \omega_\mi, x_i, a_C}$. Then
\begin{align*}
\xi_{E_B | \omega_\mi,a_C, x_i} &= \frac{1}{\P_{A_C | \omega_\mi, x_i}(a_C)} \sum_{\xvec} \P_{\Xvec | \omega_\mi, x_i}(\xvec) \sqrt{\rho} \overline{A_{\xvec}^{a_C}} \sqrt{\rho}\\ 
&= \frac{1}{\P_{A_C | \omega_\mi, x_i}(a_C)} \sqrt{\rho} \left( \sum_{\xvec} \P_{\Xvec | \omega_\mi, x_i}(\xvec)  \overline{A_{\xvec}^{a_C}} \right) \sqrt{\rho} \\
&= \frac{1}{\P_{A_C | \omega_\mi, x_i}(a_C)} \sqrt{\rho} \overline{A_{\omega_\mi, x_i}^{a_C}} \sqrt{\rho}.
\end{align*}
Similarly,
\begin{align*}
\xi_{E_B | \omega,a_C} = \frac{1}{\P_{A_C | \omega}(a_C)} \sqrt{\rho} \overline{A_{\omega}^{a_C}} \sqrt{\rho}.
\end{align*}
For all $\omega$, $x_i$, $a_C$, define the following (unnormalized) states:
\begin{align}
	\ket{\Gamma_{\omega_\mi,x_i,a_C}} = S_{\omega_\mi,x_i,a_C} \otimes \Id \, \ket{\psi} \qquad \qquad \qquad \ket{\Gamma_{\omega,a_C}} = S_{\omega,a_C} \otimes \Id \, \ket{\psi}
\end{align}
where the $S$ operators were defined in Section~\ref{sec:quantum_setup}. Let $\gamma_{\omega_\mi,x_i,a_C} = (\P_{A_C | \omega_\mi, x_i}(a_C))^{1/2} = \| \ket{\Gamma_{\omega_\mi,x_i,a_C}} \|$ and $\gamma_{\omega,a_C} = (\P_{A_C | \omega}(a_C))^{1/2} = \| \ket{\Gamma_{\omega,a_C}} \|$ denote their norms. We will write
$$
	\ket{\wt{\Gamma}_{\omega_\mi,x_i,a_C}} = \gamma_{\omega_\mi,x_i,a_C}^{-1} \ket{\Gamma_{\omega_\mi,x_i,a_C}} \qquad \qquad \qquad \ket{\wt{\Gamma}_{\omega,a_C}} = \gamma_{\omega,a_C}^{-1} \ket{\Gamma_{\omega,a_C}}
$$
to denote the normalized states.

For notational convenience we will suppress mention of $\omega_\mi$ and $z = (a_C,b_C)$, and implicitly carry them around. Thus, for example, when we write $\ket{\Gamma_{x_i}}$ and $\ket{\Gamma_{\omega_i}}$, we implicitly mean $\ket{\Gamma_{\omega_\mi,x_i,a_C}}$ and $\ket{\Gamma_{\omega,a_C}}$, respectively.

Fix $x_i$, and consider the following:
\begin{align}
	&\| \ket{\wt{\Gamma}_{x_i}} - \ket{\wt{\Gamma}_{\omega_i}} \|^2 \notag \\
	&= \left ( \bra{\wt{\Gamma}_{x_i}} - \bra{\wt{\Gamma}_{\omega_i}} \right) \left ( \ket{\wt{\Gamma}_{x_i}} - \ket{\wt{\Gamma}_{\omega_i}} \right) \notag \\
	&= \bra{\psi} (\gamma_{x_i}^{-1} S_{x_i} - \gamma_{\omega_i}^{-1} S_{\omega_i})^\dagger (\gamma_{x_i}^{-1} S_{x_i} - \gamma_{\omega_i}^{-1} S_{\omega_i}) \otimes \Id \ket{\psi} \notag  \\
	&= \Tr \left ( \sqrt{\rho} (\gamma_{x_i}^{-1} S_{x_i} - \gamma_{\omega_i}^{-1} S_{\omega_i})^\dagger (\gamma_{x_i}^{-1} S_{x_i} - \gamma_{\omega_i}^{-1} S_{\omega_i}) \sqrt{\rho} \right) \qquad \text{(Ando's Identity)}  \notag  \\
	&= \| \gamma_{x_i}^{-1} S_{x_i} \sqrt{\rho} - \gamma_{\omega_i}^{-1} S_{\omega_i} \sqrt{\rho} \|_F^2. \notag  \\
\intertext{Next we use the Powers-St\o rmer inequality~\cite{powers1970free}, which states that for positive semidefinite operators $A, B$, we have $\|A - B \|_F^2 \leq \|A^2 - B^2 \|_1$. Since $S_{x_i} \sqrt{\rho}$ and $S_{\omega_i} \sqrt{\rho}$ are by construction are positive semidefinite, the above is bounded by
}
	&\leq \| \gamma_{x_i}^{-2} S_{x_i} \rho S_{x_i}^\dagger - \gamma_{\omega_i}^{-2} S_{\omega_i} \rho S_{\omega_i}^\dagger \|_1. 	\label{eq:trace_bound} \\
\intertext{We can write $S_{x_i} \rho S_{x_i}^\dagger = U_{x_i} (A_{x_i})^{1/2} \rho (A_{x_i})^{1/2} U_{x_i}^\dagger = \sqrt{\rho} A_{x_i} \sqrt{\rho}$
and  $S_{\omega_i} \rho S_{\omega_i}^\dagger = \sqrt{\rho} A_{\omega_i} \sqrt{\rho}$. Next we observe that for any square matrix $A$, $\| A \|_1 = \| \overline{A} \|_1$, where $\overline{A}$ denotes the entry-wise complex conjugate in some basis. By taking the complex conjugate with respect to the basis that diagonalizes $\rho$, we have that~\eqref{eq:trace_bound} is equal to}
	&\| \gamma_{x_i}^{-2} \sqrt{\rho} \overline{A_{x_i}} \sqrt{\rho} - \gamma_{\omega_i}^{-2} \sqrt{\rho} \overline{A_{\omega_\mi}} \sqrt{\rho} \|_1.
	\label{eq:trace_bound1}
\end{align}
We see that~\eqref{eq:trace_bound1}, averaged over $i, \omega, a_C$ and $x_i$ is exactly the quantity bounded in~\eqref{eq:trace_bound0}. Applying Jensen's inequality, we have
\begin{align}
	\delta^{1/4} &\geq \Ex_i  \Ex_{\Omega A_C X_i} \| \ket{\wt{\Gamma}_{\omega_\mi, x_i, a_C}} - \ket{\wt{\Gamma}_{\omega,a_C}} \| \\
	&\geq \Ex_i  \Ex_{\Omega A_C X_i} \| \ketbra{\wt{\Gamma}}{\wt{\Gamma}}_{\omega_\mi, x_i, a_C} - \ketbra{\wt{\Gamma}}{\wt{\Gamma}}_{\omega, a_C} \|_1
	\label{eq:10}
\end{align}
where in the second line we used Fact~\ref{fact:trace_vs_euclid}, and we write $\ketbra{\wt{\Gamma}}{\wt{\Gamma}}_{\omega_\mi, x_i, a_C}$ instead of  \\ $\ketbra{\wt{\Gamma}_{\omega_\mi, x_i, a_C}}{\wt{\Gamma}_{\omega_\mi, x_i, a_C}}$ to save space.

Define the cq-states
$$
	\Phi_{\Omega X_i E_A E_B A_C}^i = \sum_{\omega, a_C, x_i} \P_{\Omega A_C X_i}(\omega, a_C, x_i) \,\, \ketbra{\omega x_i}{\omega x_i} \otimes \ketbra{\wt{\Gamma}}{\wt{\Gamma}}_{\omega_\mi, x_i, a_C} \otimes \ketbra{a_C}{a_C}
$$
and
$$
	\what{\Phi}_{\Omega X_i E_A E_B A_C}^i = \sum_{\omega, a_C, x_i} \P_{\Omega A_C X_i}(\omega, a_C, x_i) \,\, \ketbra{\omega x_i}{\omega x_i} \otimes \ketbra{\wt{\Gamma}}{\wt{\Gamma}}_{\omega, a_C} \otimes \ketbra{a_C}{a_C}
$$
so that the bound in~\eqref{eq:10} is equivalent to
\begin{align}
	\Ex_i \left \| \Phi_{\Omega X_i E_A E_B A_C}^i - \what{\Phi}_{\Omega X_i E_A E_B A_C}^i \right \|_1 \leq \delta^{1/4}
\end{align}
We define the quantum operation $\E$ acting on registers $\Omega E_B$ as follows: for all $\omega$ and density matrices $\tau$,
$$
	\E: \ketbra{\omega}{\omega} \otimes \tau \mapsto \ketbra{\omega}{\omega} \otimes \sum_{b_C} T_{\omega,b_C} \tau T_{\omega,b_C}^\dagger \otimes \ketbra{b_C}{b_C}.
$$

In other words, the quantum operation $\E$ will, controlled on $\Omega$, apply the measurement corresponding to the $T_{\omega,b_C}$ operators (defined in Section~\ref{sec:quantum_setup}) to the $E_B$ part of the state, and save the measurement outcomes in an ancilla register.

The operation $\E$ is an isometry, so we have that 
\begin{align}
	\Ex_i \left \| \E \left ( \Phi_{\Omega X_i E_A E_B A_C}^i \right) - \E \left ( \what{\Phi}_{\Omega X_i E_A E_B A_C}^i \right) \right \|_1 \leq \delta^{1/4}.
\end{align}
Let us examine what happens when we apply $\E$ to $\Phi_{\Omega X_i E_A E_B A_C}^i$:
\begin{align*}
	&\E\left (\Phi_{\Omega X_i E_A E_B A_C}^i \right ) \\
	&= \Ex_{\Omega A_C X_i} \ketbra{\omega x_i}{\omega x_i} \otimes \sum_{b_C} T_{\omega,b_C} \ketbra{\wt{\Gamma}_{\omega_\mi, x_i, a_C}}{\wt{\Gamma}_{\omega_\mi, x_i, a_C}} T_{\omega, b_C}^\dagger  \otimes \ketbra{a_C b_C}{a_C b_C} \\
	&= \Ex_{\Omega X_i} \sum_{a_C} \P_{A_C | \omega, x_i}(a_C) \ketbra{\omega x_i}{\omega x_i} \otimes \sum_{b_C} \frac{T_{\omega,b_C} \ketbra{\Gamma_{\omega_\mi, x_i, a_C}}{\Gamma_{\omega_\mi, x_i, a_C}} T_{\omega, b_C}^\dagger}{\P_{A_C | \omega_\mi, x_i}(a_C)}  \otimes \ketbra{a_C b_C}{a_C b_C} \\
	&= \Ex_{\Omega X_i} \ketbra{\omega x_i}{\omega x_i} \otimes \sum_{a_C, b_C} T_{\omega,b_C} \ketbra{\Gamma_{\omega_\mi, x_i, a_C}}{\Gamma_{\omega_\mi, x_i, a_C}} T_{\omega, b_C}^\dagger  \otimes \ketbra{a_C b_C}{a_C b_C}
\end{align*}
where in the second equality we used that the normalization of $\ketbra{\wt{\Gamma}}{\wt{\Gamma}}$ is equal to $\P_{A_C | \omega_\mi, x_i}(a_C)$, and that $\P_{A_C | \omega, x_i}(a_C) = \P_{A_C | \omega_\mi, x_i}(a_C)$. Similarly, we have that
$$
	\E\left (\what{\Phi}_{\Omega X_i E_A E_B A_C}^i \right ) = \Ex_{\Omega X_i} \ketbra{\omega x_i}{\omega x_i} \otimes \sum_{a_C, b_C} T_{\omega,b_C} \ketbra{\Gamma_{\omega, a_C}}{\Gamma_{\omega, a_C}} T_{\omega, b_C}^\dagger  \otimes \ketbra{a_C b_C}{a_C b_C}.
$$

Define $\Lambda^i_{\Omega X_i E_A E_B A_C B_C} = \E\left (\Phi_{\Omega X_i E_A E_B A_C}^i \right )$ and $\what{\Lambda}^i_{\Omega X_i E_A E_B A_C B_C} = \E\left (\what{\Phi}_{\Omega X_i E_A E_B A_C}^i \right )$. In both these states, the event of $W_C$ is well defined: the registers $X_C Y_C$ (which are part of the dependency-breaking variable $\Omega$) and $A_C B_C$ are classical. Furthermore, we claim that the probability of the event $W_C$ in $\Lambda^i$ and $\what{\Lambda}^i$ are equal to the probability of $W_C$ in the actual repeated strategy. Let
$$
	\Pi = \sum_{\substack{x_C, y_C, a_C,b_C : \\ V(x_C,y_C,a_C, b_C) = 1}} \ketbra{x_C y_C a_C b_C} {x_C y_C a_C b_C}
$$
be the projector onto the subspace corresponding to the event $W_C$. Then for all $i$
\begin{align*}
	&\Tr \left( \Pi \Lambda^i \right ) \\
	&= \sum_{\omega, x_i} \P_{\Omega X_i}(\omega x_i) \sum_{\substack{a_C, b_C : \\ V(x_C,y_C,a_C,b_C) = 1}} \bra{\Gamma_{\omega_\mi,x_i,a_C}} T^\dagger_{\omega,b_C} T_{\omega,b_C} \ket{\Gamma_{\omega_\mi,x_i,a_C}} \\
	&= \sum_{\omega, x_i} \P_{\Omega X_i}(\omega x_i) \sum_{\substack{a_C, b_C : \\ V(x_C,y_C,a_C,b_C) = 1}} \bra{\psi}  \left ( S_{\omega_\mi,x_i, a_C} \otimes T_{\omega,b_C} \right)^\dagger \left ( S_{\omega_\mi,x_i, a_C} \otimes T_{\omega,b_C} \right) \ket{\psi} \\
	&= \sum_{\omega, x_i} \P_{\Omega X_i}(\omega x_i) \sum_{\substack{a_C, b_C : \\ V(x_C,y_C,a_C,b_C) = 1}} \bra{\psi}  \left ( \sqrt{A_{\omega_\mi,x_i}^{a_C}} \otimes \sqrt{B_\omega^{b_C}} \right)^\dagger \left ( \sqrt{A_{\omega_\mi,x_i}^{a_C}} \otimes \sqrt{B_\omega^{b_C}} \right) \ket{\psi}.
\end{align*}
Using the definitions of $A_{\omega_\mi,x_i}^{a_C}$ and $B_\omega^{b_C}$ we see that this quantity is identical to $\P(W_C)$. Similar reasoning shows that $\Tr \left( \Pi\what{\Lambda}^i \right ) = \P(W_C)$.

Let $\Lambda^i_{\Omega X_i E_A E_B A_C  B_C| W_C} = (\Pi \Lambda^i \Pi)/\P(W_C)$ and $\what{\Lambda}^i_{\Omega X_i E_A E_B A_C  B_C| W_C} = (\Pi \what{\Lambda}^i \Pi)/\P(W_C)$ denote $\Lambda^i$ and $\what{\Lambda}^i$ \emph{conditioned} on the event $W_C$. So we have
\begin{align}
	\Ex_i \left \| \Lambda^i_{\Omega X_i E_A E_B A_C  B_C| W_C} - \what{\Lambda}^i_{\Omega X_i E_A E_B A_C  B_C| W_C} \right \|_1 \leq \frac{\delta^{1/4}}{\P(W_C)}.
	\label{eq:cond_dist}
\end{align}
Let us bundle together the $\Omega$ and $A_C B_C$ registers into $R$. For all $r = (\omega,a_C,b_C)$ and $x_i$, define
$$
\ket{\Psi_{r,x_i}} = \frac{S_{\omega_\mi,x_i,a_C} \otimes T_{\omega,b_C} \ket{\psi}}{ \left \| S_{\omega_\mi,x_i,a_C} \otimes T_{\omega,b_C} \ket{\psi} \right \| } 
\qquad \qquad \ket{\Psi_{r}} = \frac{S_{\omega,a_C} \otimes T_{\omega,b_C} \ket{\psi}}{ \left \| S_{\omega,a_C} \otimes T_{\omega,b_C} \ket{\psi} \right \| }
$$
Then we see that
$$
	\Lambda^i_{R X_i E_A E_B | W_C} = \Ex_{R X_i | W_C} \ketbra{r x_i}{r x_i} \otimes \ketbra{\Psi_{r,x_i}}{\Psi_{r,x_i}}
$$
and
$$
	\what{\Lambda}^i_{R X_i E_A E_B | W_C} = \Ex_{R | W_C} \ketbra{r}{r} \otimes \Ex_{X_i | \omega} \ketbra{x_i}{x_i} \otimes \ketbra{\Psi_{r}}{\Psi_{r}}.
$$
We see that $\Lambda^i_{R X_i E_A E_B | W_C}$ and $\what{\Lambda}^i_{R X_i E_A E_B | W_C}$ are both cq-states that are classical on $R X_i$ and quantum on $E_A E_B$. The inequality in~\eqref{eq:cond_dist} implies that the trace distance between the classical parts of $\Lambda^i_{W_C}$ and $\what{\Lambda}^i_{W_C}$ is at most $\delta^{1/4}/\P(W_C)$. Thus we can change the classical part of $\what{\Lambda}^i_{W_C}$ to match the classical part of $\Lambda^i_{W_C}$ by at most doubling the error:
$$
\Ex_i \left \| \Ex_{R X_i | W_C} \ketbra{r x_i}{r x_i} \otimes \left (\ketbra{\Psi_{r,x_i}}{\Psi_{r,x_i}} - \ketbra{\Psi_{r}}{\Psi_{r}} \right ) \right \|_1 \leq 2\frac{\delta^{1/4}}{\P(W_C)}.
$$
which implies that
\begin{align*}
	\Ex_i \Ex_{R X_i | W_C} \left \|\ketbra{\Psi_{r,x_i}}{\Psi_{r,x_i}} -  \ketbra{\Psi_{r}}{\Psi_{r}} \right \|_1 \leq \frac{2\delta^{1/4}}{\P(W_C)}.
\end{align*}
By Lemma~\ref{lem:classical_skew}, $\Ex_i \| \P_{\Omega_i X_i | W_C} - \P_{\Omega_i X_i} \|_1 \leq \sqrt{\delta}$. Applying that to the above, we get
\begin{align*}
	\Ex_i \Ex_{\Omega_i X_i} \left [ \Ex_{R_\mi | \omega_i, x_i, W_C} \left \|\ketbra{\Psi_{r,x_i}}{\Psi_{r,x_i}} -  \ketbra{\Psi_{r}}{\Psi_r} \right \|_1 \right ] \leq \frac{2\delta^{1/4}}{\P(W_C)} + \sqrt{\delta}
\end{align*}
where the middle expectation over $\Omega_i X_i$ is over the \emph{prior} distribution (i.e. before conditioning on the event $W_C$). Now observe that in this prior distribution, $\Omega_i$ fixes $Y_i$ with probability $1/2$, so we in fact get
\begin{align*}
	\Ex_i \Ex_{X_i Y_i} \left [ \Ex_{R_\mi| x_i, y_i, W_C} \left \|\ketbra{\Psi_{r_\mi,x_i,y_i}}{\Psi_{r_\mi,x_i,y_i}} -  \ketbra{\Psi_{r_\mi,y_i}}{\Psi_{r_\mi,y_i}} \right \|_1 \right ]  \leq \frac{4\delta^{1/4}}{\P(W_C)} + 2\sqrt{\delta}
\end{align*}
where the states $\ket{\Psi_{r_\mi,x_i,y_i}}$ were defined in Section~\ref{sec:quantum_setup}, and $\ket{\Psi_{r_\mi,y_i}}$ is $\ket{\Psi_r}$ where $r = (r_\mi,\omega_i)$ and $\omega_i$ fixes $Y_i = y_i$. Applying the Fuchs-van der Graaf inequality, we obtain a bound in terms of Euclidean distance:
\begin{align}
	\Ex_i \Ex_{X_i Y_i} \left [ \Ex_{R_\mi | x_i, y_i, W_C} \left \|\ket{\Psi_{r_\mi,x_i,y_i}} -  \ket{\Psi_{r_\mi,y_i}} \right \| \right ] \leq O \left( (\delta^{1/4}/\P(W_C) )^{1/2} \right)
	\label{eq:bob_bound}
\end{align}
Similar reasoning implies that 
\begin{align}
	\Ex_i \Ex_{X_i Y_i} \left [ \Ex_{R_\mi | x_i, y_i, W_C} \left \|\ket{\Psi_{r_\mi,x_i,y_i}} -  \ket{\Psi_{r_\mi,x_i}} \right \| \right ]  \leq O \left( (\delta^{1/4}/\P(W_C) )^{1/2} \right)
	\label{eq:alice_bound}
\end{align}
where $\ket{\Psi_{r_\mi,x_i}}$ is $\ket{\Psi_r}$ where $r = (r_\mi,\omega_i)$ and $\omega_i$ fixes $X_i = x_i$.
By triangle inequality, we have
\begin{align}
\Ex_i \Ex_{X_i Y_i} \left [ \Ex_{R_\mi | x_i, y_i, W_C} \left \|\ket{\Psi_{r_\mi,y_i}} -  \ket{\Psi_{r_\mi,x_i}} \right \| \right ]  \leq O \left( (\delta^{1/4}/\P(W_C) )^{1/2} \right).
\label{eq:alice_bob_bound}
\end{align}
Let $\eta := O \left( (\delta^{1/4}/\P(W_C) )^{1/2} \right)$. Fix $r_\mi, x_i, y_i$. Since $r_\mi$ is public, Alice knows $r_\mi, x_i$, and thus knows a classical description of the state $\ket{\Phi_{r_\mi,x_i}}$. Similarly, Bob knows a classical description of the state $\ket{\Phi_{r_\mi,y_i}}$. By the Quantum Correlated Sampling Lemma of~\cite{DinurSV14} with parameter $\alpha = \eta^6$, there exists a dimension $d'$ that depends only on $d$ and $\alpha$, and unitaries $U_{r_\mi,x_i}$ and $V_{r_\mi,y_i}$ such that
$$
\| U_{r_\mi,x_i} \otimes V_{r_\mi,y_i} \ket{E_{dd'}} - \ket{\Psi_{r_\mi,x_i}} \ket{E_{d'}} \| \leq O(\max \{ \alpha^{1/12}, \left \| \, \ket{\Psi_{r_\mi,x_i}} - \ket{\Psi_{r_\mi,y_i}} \, \right \|^{1/6} \} ).
$$
We can average this over $i$, $x_i, y_i$, and $r_\mi$ to get that
\begin{align*}
&\Ex_i \Ex_{X_i Y_i} \left [ \Ex_{R_\mi | x_i, y_i, W_C} \| U_{r_\mi,x_i} \otimes V_{r_\mi,y_i} \ket{E_{dd'}} - \ket{\Psi_{r_\mi,x_i}} \ket{E_{d'}} \| \right ] \\
&\leq \Ex_i \Ex_{X_i Y_i} \left [ \Ex_{R_\mi | x_i, y_i, W_C} O(\max \{ \alpha^{1/12}, \left \| \, \ket{\Psi_{r_\mi,x_i}} - \ket{\Psi_{r_\mi,y_i}} \, \right \|^{1/6} \} ) \right ] \\
&\leq O(\alpha^{1/24}) \\
&= O(\eta^{1/12})
\end{align*}
where in the second inequality we used the following fact: for an nonnegative random variable $X$ with mean $\mu = \Ex X$, we can bound the expectation $\Ex \max \{ \sqrt{\mu}, X \} \leq O(\sqrt{\mu})$. Using the bound~\eqref{eq:alice_bound}, we get
$$
\Ex_i \Ex_{X_i Y_i} \left [ \Ex_{R_\mi | x_i, y_i, W_C} \| U_{r_\mi,x_i} \otimes V_{r_\mi,y_i} \ket{E_{dd'}} - \ket{\Psi_{r_\mi,x_i,y_i}} \ket{E_{d'}} \| \right ] \leq O(\eta^{1/12})
$$
as desired.


\paragraph{Acknowledgments.} This work was supported by Simons Foundation grant 360893 and National Science Foundation Grant 1218547. The author thanks both the Institute of Mathematical Sciences at the National University of Singapore, and the Weizmann Institute of Science for hospitable stays during which this research was conducted. The author also thanks Corinna Li, Mohammad Bavarian, Govind Ramnarayan, and anonymous referees for helpful feedback and discussions.

\bibliography{qverbitsky}

\end{document}